\documentclass[a4paper]{amsart}
\usepackage{amssymb,amscd,amsmath}

\usepackage{tikz}
\usetikzlibrary{matrix}

\theoremstyle{definition}
\newtheorem{formula}{*}[section]
\newtheorem{definition}[formula]{Definition}
\newtheorem{corollary}[formula]{Corollary}
\newtheorem{remark}[formula]{Remark}
\newtheorem{lemma}[formula]{Lemma}
\newtheorem{theorem}[formula]{Theorem}
\newtheorem{example}[formula]{Example}
\newtheorem{proposition}[formula]{Proposition}

\newcommand{\forme}[1]{}

\begin{document}

\title{weighted Hamming metric structures}

\author{Bora Moon}
\address{Department of Mathematics, Postech, 77 Cheongam-Ro. Nam-Gu. Pohang. Gyeongbuk, Korea.}

\email{mbr918@postech.ac.kr}

\date{\today}

\maketitle

\begin{abstract}
A weighted Hamming metric is introduced in \cite{BS13} and it showed that the binary generalized Goppa code is a perfect code in some weighted Hamming metric.
In this paper, we study the weight structures which admit the binary Hamming code and the extended binary Hamming code to be perfect codes in the weighted Hamming metric. And, we also give some structures of a 2-perfect code and how to construct a 2-perfect code in some weighted Hamming metric.

\end{abstract}

\section{Introduction}

There are attempts to consider coding theory not only in the Hamming metric but also in other metrics (cf.
\cite{HKP17},\cite{JK04},\cite{BS13} ). Due to these attempts, we have many types of perfect codes. \cite{BS13} S.Bezzateev and N.Shekhunova considered a perfect code in the weighted Hamming metric to apply codes on the channels with nonuniform distribution of errors. And they gave some basic properties of a perfect code in the weighted Hamming metric and they showed that a binary generalized Goppa code is a perfect code in the weighted Hamming metric. So, we study to know if the (extended) binary Hamming code is perfect in some weighted Hamming metric and in which weight structures it becomes perfect. And then we consider a particular weighted Hamming metric and we will induce how to construct a 2-perfect code in the metric and its properties. So we can have new types of a perfect code.\\

 Let $\mathbb{F}_2$ be a finite field of order two and $\mathbb{F}^n_2$ a vector space of binary $n$-tuples. 

\begin{definition}\label{definition1.1}\cite{BS13} Let $\pi_i\in \mathbb{N}$ and $ \pi=(\pi_1,\cdots,\pi_n)$ be a weight of position $i$ and a vector of weights, respectively.

The $\pi$-weight $w_\pi$ of a vector $\mathbf{x}=(x_1,\cdots,x_n)$ of $\mathbb{F}_2^n$ is defined by a function
\[w_\pi(x)=\sum_{i=1}^n \pi_i\cdot x_i.\]

And the $\pi$-distance $d_\pi(\mathbf{x},\mathbf{y})$ between vectors $\mathbf{x}=(x_1,\cdots,x_n)$ and $\mathbf{y}=(y_1,\cdots,y_n)$ of $\mathbb{F}_2^n$ is defined by $d_\pi(\mathbf{x},\mathbf{y})=w_\pi(\mathbf{x-y})$.

For a subset $C\subset \mathbb{F}_2^n$, we call $C$ a $\pi$-code of length $n$.
\end{definition}

For a given vector of weights $\pi$, a $\pi$-metric $d_\pi$ is  also called a weighted Hamming metric on $\mathbb{F}_2^n$.

Further, if $\pi=(1,1,\cdots,1)$, then $w_H(\mathbf{x})=w_\pi(\mathbf{x})$ ($d_H(\mathbf{x})=d_\pi(\mathbf{x})$) is called a Hamming weight of a vector $\mathbf{x}$ (Hamming distance of a vector $\mathbf{x}$).\\

Let  $\mathbf{x}$ be a vector of $\mathbb{F}_2^n$ and $r$ a nonnegative integer. The $\pi$-sphere with center $\mathbf{x}$
and radius $r$ is defined as the set
\[S_\pi(\mathbf{x};r)=\{\mathbf{y}\in \mathbb{F}_2^n | d_\pi(\mathbf{x,y})\le r \}.\]   
    
The support of a vector $\mathbf{x}$ of $\mathbb{F}^n_2$ is the set of non-zero coordinate positions of a vector $\mathbf{x}$, so the size of the support of $\mathbf{x}$ is also $w_H(\mathbf{x})$. Throughout this paper, we identify a vector $\mathbf{x}$ of $\mathbb{F}^n_2$ with its support.\\

\begin{lemma}\label{lemma1.2}
For a given vecotor of weights $\pi$, a weight metric $d_\pi$ is a metric on $\mathbb{F}_2^n$.

\end{lemma}

\begin{definition}\label{definition1.3}\cite{HKP17}\label{std} Let $\pi$ be a vector of weights and $C$ a $\pi$-code of length $n$.

We say that a code $C$ is an $r$-perfect $\pi$-code if the union of the spheres $S_\pi(\mathbf{c};r)$ centered at $\mathbf{c}\in C$ equals to $\mathbb{F}_2$ and the spheres are mutually disjoint.
\end{definition}

Denote the set of coordinates of $\pi$-weight $i$ and the size of $X_i$ by $X_i$ and $x_i$, respectively for $i=1,2,\cdots$.
We can consider a vector of weights $\pi$ corresponds to the set $\{X_i~|~i=1,2,3,\cdots \}.$

\begin{proposition}\cite{HKP17}\label{proposition1.4} Let $\pi$ be a vector of weights and $C$ an $[n,k]$ binary linear $\pi$-code. Then, $C$ is an $r$-perfect $\pi$-code if and only if the following two conditions are satisfied:
\begin{enumerate}
\item (\textit{The sphere packing condition}) $|S_\pi (\textbf{0};r)|$ = $2^{n-k}$.
\item (\textit{The partition condition}) for any non-zero codeword $\mathbf{c}$ and any partition $\{\mathbf{x,y}\}$ of $\mathbf{c}$, either $w_\pi(\mathbf{x})$ $\ge$ $r+1$ or
$w_\pi(\mathbf{y})$ $\ge$ $r+1$

\end{enumerate}

\end{proposition}


\section{perfect $\pi$-codes $\mathcal{H}_m$ and $ \widetilde{\mathcal{H}}_m$}

\begin{definition}\label{definition1.5}
Let  $\widetilde{H}_m$ be an $(m+1)\times 2^m$ binary matrix whose first row is the all one vector of length $2^m$ and the remaining $m$ rows of $\widetilde{H}_m$ form a $m\times 2^m$ submatrix whose $i$-th column corresponds to the 2-adic representation of $i$. By deleting the first row and the first column of  $\widetilde{H}_m$, we have an $m\times (2^m-1)$ binary matrix  $H_m$. 
Let $\mathcal{H}_m = [n=2^m-1,2^m-1-m,3]$ and $\widetilde{\mathcal{H}}_m$ = $[n=2^m,2^m-1-m,4]$ $(m\ge 2)$ be the binary Hamming code with the parity check matrix $H_m$ and the extended binary Hamming code with the parity check matrix $\widetilde{H}_m$, respectively. 
\end{definition}

\subsection{2-perfect $\pi$-codes $\mathcal{H}_m$ and $\widetilde{\mathcal{H}}_m$}

By proposition~\ref{proposition1.4}, if a binary Hamming code $\mathcal{H}_m(m\ge2)$ is a 2-perfect $\pi$-code for some vector of weights $\pi$, then
\[2^m=|S_\pi(\mathbf{0};2)|=1+x_1+\binom{x_1}{2}+x_2.\]
Thus, $x_1+\binom{x_1}{2}+x_2=2^m-1$.

\begin{theorem}
For $x_1,x_2\in \mathbb{N}\cup\{0\}$, there is a vector of weights $\pi$ admitting a binary Hamming code $\mathcal{H}_m$ to be a 2-perfect $\pi$-code if the following conditions hold.
\[1+\binom{x_1-1}{1}+\binom{x_1-1}{2}+\binom{x_1-1}{3}<2^m \cdots(1)\]
\[x_2=2^m-1-x_1-\binom{x_1}{2}\ge 0 \cdots(2)\]

\end{theorem}

\begin{proof}
By $(1)$ and the proof in \cite{MS83} the Gilbert-Varshamov Bound, we can take a $x_1$-subset $X_1$ of $[n]$ which satisfies that no $4$ or fewer elements (with respect to a 2-adic representation) are linearly dependent (with respect to a 2-adic representations).
To be a 2-perfect $\pi$-code, for any $\mathbf{c}\in \mathcal{H}_m$, for any partition $\{\mathbf{x,y}\}$ of $\mathbf{c}$, $w_\pi(\mathbf{x})\ge3$ or $w_\pi(\mathbf{y})\ge3$.

Let $\phi:\binom{X_1}{2}\rightarrow [n]$ be defined by $\phi(\mathbf{x})=\mathbf{c-x}$ for $\mathbf{x}\in \binom{X_1}{2}$ where $\mathbf{c}\in \mathcal{H}_m$ is of Hamming weight $3$ which contains $\mathbf{x}$.

Then, $\phi$ must be well-defined, 1-1 and onto. Otherwise, let $\mathbf{y}:=\phi(\mathbf{x_1})=\phi(\mathbf{x_2})$. 

Then, $\mathbf{x_1}\cup\mathbf{y}=\mathbf{x_2}\cup\mathbf{y}\in \mathcal{H}_m$ and we get $\mathbf{x_1}\cup \mathbf{x_2}\in \mathcal{H}_m$. It is a contradiction to the choice of $X_1$.
We have to take $\cup_{i\ge 3}X_i:=\phi(\binom{X_1}{2})$ and $X_2=[n] - (X_1\cup~ \phi(\binom{X_1}{2})).$ Then, we get a vector of weights $\pi$ admitting a binary Hamming code $\mathcal{H}_m$ to be a 2-perfect $\pi$-code.

\end{proof}

Let's see an example that shows how to find a vector of weights $\pi$ admitting a binary Hamming code $\mathcal{H}_4$ to be a 2-perfect $\pi$-code.

\begin{example}
When $m=4$, take $X_1= \{1,2,4,8,15\}$. Then,
\[\binom{X_1}{2}=\{\{1,2\}, \{1,4\}, \{1,8\}, \{1,15\}, \{2,4\}, \{2,8\}, \{2,15\}, \{4,8\}, \{4,15\}, \{8,15\}    \}\]
and it must be $X_2=\emptyset$, $\cup_{i\ge 3}X_i=[15]-X_1$. Then, a binary Hamming code $\mathcal{H}_4$ becomes a 2-perfect $\pi$-code.

\end{example}

Now we will find out for which $m$ the extended binary Hamming code $\widetilde{\mathcal{H}}_m(m\ge 2)$ becomes a 2-perfect code in the weighted Hamming metric.

\begin{lemma}\label{lemma2.1} Let $\pi$ be a vector of weights. If the extended binary Hamming code $\widetilde{\mathcal{H}}_m(m\ge2)$ is a 2-perfect $\pi$-code, then there is no coordinate in $\{1,2,\cdots,2^m\}$ whose $\pi$-weight is bigger than two.

\end{lemma}

By the above lemma, $\{1,2,\cdots,2^m\}=X_1\cup X_2$. Then, it's easy to see that 
$$2^{m+1}=|S_\pi(0;2)|=1+2^m+\binom{x_1}{2}.$$

If $x_1=2k$, $k=\frac{1+\sqrt{2^{m+3}-7}}{4}$ $\in \mathbb{N}$ and if $x_1=2k+1$, $k=\frac{-1+\sqrt{2^{m+3}-7}}{4}$ $\in \mathbb{N}$.

Then, $2^{m+3}-7$ = $(4s\pm 1)^2$ for some $s\in \mathbb{N}$.

\begin{theorem}\cite{N61}\label{std}
\textbf{[Nagell's equation]} For $x,n \in \mathbb{N}$, the equations $x^2+7=2^n$ has only solutions given by  $x=1,3,5,11,181$ corresponding to $n=3,4,5,7,15$.

\end{theorem}

Due to the solutions of Nagell's equation, we get the following values.

\begin{center}
\begin{tabular}{c|c|c|c|c}

$(4s\pm1)$ & $m+3$ & $s$ & $m$ & $x_1$\\
\hline
\hline
1 & 3 & 0 & 0 & 1\\

3 & 4 & 1 & 1 & 2\\

5 & 5 & 1 & 2 & 3\\

11 & 7 & 3 & 4 & 6\\

181 & 15 & 45 & 12 & 91\\

\end{tabular}
\end{center}

\begin{enumerate}

\item When $m=2$, $C=\{0000,1111\}$ and $X_1=\{i,j,k\}$, $X_2=\{l\}$ where $\{i,j,k,l\}=\{1,2,3,4\}$.

\item When $m=4$ ($x_1=6$, $x_2=10$), there is some vector of weights $\pi$ admitting $\widetilde{\mathcal{H}}_{4}$ to be a 2-perfect $\pi$-code. To get a desired vector of weights $\pi$, we need to show that $X_1$ is in $\widetilde{\mathcal{H}}_{4}$.

Say  $X_1=\{\alpha,\beta,\gamma,\delta,\epsilon,\eta\}$ and let the extended binary Hamming code $\widetilde{\mathcal{H}_4}$ be a 2-perfect $\pi$-code for some vector of weights $\pi$.
Let $\mathcal{X}$ be a maximal set among the sets of 3-subsets of $X_1$ for which any two 3-subsets intersect.

Define $\Phi$ : $\mathcal{X}$ $\rightarrow$ $X_2$ by for $A\in \mathcal{X}$, $A\cup \Phi(A)$ $\in$ $\widetilde{\mathcal{H}}_{4}$ $\cdots (*)$. 

For the perfectness, $\Phi$ must be well-defined, 1-1 and onto. 

If $\Phi(A)=\Phi(A')$, then $B:=(A\backslash A')\cup (A'\backslash A)\in \widetilde{\mathcal{H}}_{4}$, $|B|\le 4$ and $0\le w_\pi(B) \le 4$. So, $A=A'$.

 Let $\mathcal{X}$ and $\mathcal{X'}$ be maximal ones with $\{\alpha,\beta,\gamma\}$ $\in \mathcal{X}$, $\{\delta,\epsilon,\eta\}$ $\in \mathcal{X}'$ and $\Phi_1$ and $\Phi_2$ are maps defined on $\mathcal{X}$ and $\mathcal{X}'$, respectively as in $(*)$. 
 
 Put $x:=\Phi_1(\{\alpha,\beta,\gamma\})$ and $y:=\Phi_2(\{\delta,\epsilon,\eta\})$. If $x\not =y $, then there is $\widetilde{A}\in \mathcal{X}'$ such that $\{\alpha,\beta,\gamma\} \cap \widetilde{A} \not = \emptyset$, 
 $\Phi_2(\widetilde{A})=x$. So, $\{\alpha,\beta,\gamma,x\}$, $\{x\}\cup \widetilde{A}$ $\in \widetilde{\mathcal{H}}_{4}$. Then, there is $ \mathbf{c}  \in \widetilde{\mathcal{H}}_{4}$
 with $w_\pi(\mathbf{c})\le 4$. So, $x=y$ and $X_1 \in \widetilde{\mathcal{H}}_{4}$. 
 
 Take a 6-element codeword of $\widetilde{\mathcal{H}}_{4}$ as $X_1$ and set $X_2=\{1,2,\cdots,2^m\}$$-X_1$.

\item When $m=12$ ($x_1=91$, $x_2=4005$), $\widetilde{\mathcal{H}}_{12}$ cannot be a 2-perfect $\pi$-code for any vector of weights $\pi$.

Suppose that $\widetilde{\mathcal{H}}_{12}$ is a 2-perfect $\pi$-code for some vector of weights $\pi$. For $a\in X_1$, set $T_a$ a set of all 2-subsets of $X_1$$-\{a\}$.

Define a map $\tau_a $ : $T_a$ $\rightarrow$ $X_2$ by for $\{\alpha,\beta\}\in T_a$, $\{a,\alpha,\beta,\tau_a(\{\alpha,\beta\})\}$  $\in \widetilde{\mathcal{H}}_{12}$. Then, $\tau_a$ is well-defined, 1-1 and onto.

If not, for $\{\alpha,\beta\}$, $\{\alpha',\beta'\}\in T_a$, let $y:=\tau_a(\{\alpha,\beta\})=\tau_a(\{\alpha',\beta'\})$. Then, $\{a,\alpha,\beta,y\}$, $\{a,\alpha',\beta',y\}$  $\in \widetilde{\mathcal{H}}_{12}$ but there is no $\mathbf{c} \in  \widetilde{\mathcal{H}}_{12}$ such that $\mathbf{c} \not=0$, $w_\pi(c)\le 4$. 

Thus, $\tau_a$ is 1-1 and also onto because of $|T_a|=|X_2|=4005$.

Fix $y\in X_2$ and choose $\alpha_1\in X_1$. Then, by a map $\tau_{\alpha_1}$, there exist $\alpha_2$, $\alpha_3$
$\in X$$-\{\alpha_1\}$ such that $\{\alpha_1,\alpha_2,\alpha_3,y\}\in \widetilde{\mathcal{H}}_{12}$.

Choose $\alpha_4\in X_1$$-\{\alpha_i\}_{i=1}^3$. Then, by a map $\tau_{\alpha_4}$, there exist $\alpha_5, \alpha_6\in X_1$$-\{\alpha_i\}_{i=1}^4$ such that $\{\alpha_4,\alpha_5,\alpha_6,y\}\in \widetilde{\mathcal{H}}_m$

Continue this procedure until we have $\alpha_1, \alpha_2, \cdots, \alpha_{90}$. 

Note that all $\alpha_i$'s are distinct and $x_1=91$. For $\alpha_{91}\in X$ $\backslash$ $\{\alpha_i\}_{i=1}^{90}\not= \emptyset$, the preimage $t^{-1}_{\alpha_{91}}(y)$ is contained in $\{\alpha_i\}_{i=1}^{90}$. It's impossible.

So, we get the following theorem.

\end{enumerate}

\begin{theorem}
The extended binary Hamming code $\widetilde{\mathcal{H}}_m$($m\ge 2$) is a 2-perfect $\pi$-code for some vector of weights $\pi$ {\it if and only if} $m$ is 2 or 4.

\end{theorem}


\subsection{3-perfect $\pi$-code $\widetilde{\mathcal{H}}_m$}

Put $M_3$ = $\{\mathbf{x} \subset \{1,2,\cdots, 2^m\}$ $|$ $\mathbf{x}$ is a 3-subset with $w_\pi(\mathbf{x})=3\}$. 

\cite{JK04} J.Y.Hyun proved the following theorem.

\begin{theorem}\label{theorem2.7}

There is a bijection from $M_3$ to $Y:=\cup_{i\ge4}X_i$.

\end{theorem}

\begin{theorem}\label{theorem2.8}
For a given vector of weights $\pi$, let the extended binary Hamming code $\widetilde{\mathcal{H}}_m$($m\ge 2$) be a 3-perfect $\pi$-code. 

Then, $1 \le x_1 \le 3$.
Further, if $1 \le x_1\le 3$,  $\widetilde{\mathcal{H}}_m$($m\ge 2$) can be a 3-perfect code in some weighted Hamming metric.

\end{theorem}

\begin{proof}
By Theorem~\ref{theorem2.7}, $2^m=x_1+x_2+x_3+\binom{x_1}{3}$ and 
\[ |S_\pi(0;3)|=2^{m+1}=1+x_1+x_2+\binom{x_1}{2}+x_3+\binom{x_1}{3}+x_1 x_2.\]  So, $2^m=1+\binom{x_1}{2}+x_1 x_2.$

Note that
$$1+\binom{x_1}{2}+x_1 x_2= \binom{x_1}{3} +x_1+x_2+x_3.\cdots(3)$$

If $x_1=1$, then $x_2=2^m-1$ and $x_3=x_4=\cdots=0$. So, set $X_1=\{i\}$ and $X_2=\{1,2,\cdots, 2^m\}-\{i\}$ for $i\in \{1,2,\cdots, 2^m\}$. It makes $\widetilde{\mathcal{H}}_m$ a 2-perfect $\pi$-code.

If $x_1=2$, then $x_2=x_3=2^{m-1}-1$ and $x_4=x_5=\cdots=0$. For any $\alpha,\beta \in \{1,2,\cdots, 2^m\}$, set $X_1=\{\alpha,\beta\}$. For a codeword $\{\alpha,\beta,x_i,y_i\}\in \widetilde{\mathcal{H}}_m$ for $i=1,2,\cdots,2^{m-1}-1$,  set $x_i\in X_2$ and $y_i \in X_3$. It also makes $\widetilde{\mathcal{H}}_m$ a 2-perfect $\pi$-code.

When $x_1=3$, arbitrary take $X_1=\{\alpha,\beta,\gamma\}$ and then we have to put $Y=\{\delta\}$ where $\{\alpha,\beta,\gamma,\delta\}\in \widetilde{\mathcal{H}}_m$.

For $\{\alpha,\beta ,\nu_i,\eta_i\}\in \widetilde{\mathcal{H}}_m - \{\alpha,\beta,\gamma,\delta\}$ for $i=1,2,\cdots$,$2^{m-1}-2$, we have to set $X_2=\{\nu_1,\nu_2,\cdots, \nu_{x_2}\}$ and $X_3=\{\eta_1,\eta_2,\cdots,\eta_{x_2}\}\cup \{\nu_{x_2+1},\nu_{x_2+2},\cdots,\eta_{x_2+1},\eta_{x_2+2},\cdots\}$.

For any $\nu_i\in X_2$ $i=1,2,\cdots, x_2$, there is $\{\alpha,\gamma,\nu_i,\zeta_i\}$ in $\widetilde{\mathcal{H}}_m$. But $\zeta_i \not \in $ $\{\eta_1,\eta_2,\cdots,\eta_{x_2}\}$ for all $i=1,2,\cdots, x_2$.

And for any $\nu_i\in X_2$ $i=1,2,\cdots, x_2$, there is $\{\beta,\gamma,\nu_i,\xi_i\}$ in $\widetilde{\mathcal{H}}_m$. But $\xi_i \not \in $ $\{\eta_1,\eta_2,\cdots,\eta_{x_2}\}\cup \{\zeta_1,\zeta_2,\cdots,\zeta_{x_2}\}$ for all $i=1,2,\cdots, x_2$.

So, $X_3=\{\eta_1,\eta_2,\cdots,\eta_{x_2}\}\cup \{\zeta_1,\zeta_2,\cdots,\zeta_{x_2}\} \cup \{\xi_1,\xi_2,\cdots,\xi_{x_2}\}\cup X_3'$ for some subset $X_3'$ of $X_3$.

Thus, it must be $x_3 \ge 3 x_2$. Using this inequality and $(3)$, we have $m=2$, $x_1=3$, $x_2=x_3=0$ and $|Y|=1$. Taking $X_1=\{i,j,k\}$, $Y=\{l\}$ where $\{i,j,k,l\}=\{1,2,3,4\}$, we have a 3-perfect $\pi$-code $\widetilde{\mathcal{H}}_2$.

When $x_1\ge 4$, take $X_1=\{\alpha,\beta_1,\beta_2,\cdots,\beta_{x_1-1}\}$. By the similar arguments as above, we have $x_3\ge (x_1-1)x_2$. By $(3)$, we get $1+\binom{x_1}{2}+x_1 x_2 \ge \binom{x_1}{3}+x_1+x_2+(x_1-1)x_2$. 
Then, $x_1 \le 1$ or $2\le x_1 \le 3$. So, $x_1\ge 4$ is impossible.

\end{proof}


\section{binary 2-perfect $\pi$-codes}

From now on, let $m$ be an integer with $1\le m\le n$ and assume a vector of weights $\pi=(\pi_1,\cdots,\pi_n)$ with

\begin{displaymath}
\pi_i=\left \{\begin{array}{ll}
1, & \textrm{if $1\le i\le m$}\\
 2, & \textrm{otherwise}.
\end{array} \right.
\end{displaymath}

\subsection{Construction of a binary 2-perfect linear $\pi$-code}
Suppose $1+n+\binom{m}{2}=2^t$ for some positive integer $t$. For positive integers $t$ and $m$ with $2^t-1-m-\binom{m}{2}>0$, let $\mathcal{H}$ be a family of all $t\times (2^t-1-\binom{m}{2})$ matrices having the following properties:

\begin{enumerate}
\item It contains only nonzero vectors as column and no repeated column.

\item The sum of any distinct two columns among the first $m$ columns do not appear in the last $n-m$ columns.

\item The sum of any distinct two or three columns among the first $m$ columns do not appear in the first $m$ columns.

\end{enumerate}

So we get the following theorem.

\begin{theorem}

Let $H$ be a parity check matrix of a 2-perfect linear $\pi$-code. Then, $H\in \mathcal{H}$. 
Furthermore, any code having $H\in \mathcal{H}$ as a parity check matrix is a 2-perfect linear $\pi$-code.

\end{theorem}

\begin{remark}
Since $n-m=2^t-1-m-\binom{m}{2}$, if the first $m$ columns of a parity check matrix $H$ are decided, then the matrix $H$ is completely determined. Furthermore, if $m\le 4$, every matrix in $\mathcal{H}$ is essentially the same.

\end{remark}

\begin{example}
when $m=2$ and $t=3$, we can make a parity check $H$ as follow:

\[ H=
\left[{\begin{array}{cccccc}

1 & 0 & 0 & 1& 0 &1\\
0 & 1 & 0 & 0& 1 &1\\
0 & 0 & 1 & 1& 1 &1\\

\end{array} } \right] .\]

Then, $C=\{000000,001111,101100,100011,011010,010101,111001,110110\}$ is a 2-perfect linear $\pi$-code.

\end{example}

\subsection{Properties of a binary 2-perfect $\pi$-code}

From \cite{O02}, for $C\subset \mathbb{F}_2^n$, a code $C$ can be corresponding to a polynomial 
\[C(x_1,x_2,\cdots,x_n)=\sum_{c\in C}x_1^{c_1}\cdots x_n^{c_n}\]
and
 \[\{y_{\mathbf{t}}(x_1,x_2,\cdots,x_n)=\frac{1}{2^n}\prod_{i=1}^n(1-x_i)^{t_i}(1+x_i)^{1-t_i}|\mathbf{t}\in \mathbb{F}_2^m\}\]

forms a basis for $\mathbb{R}[x_1,x_2,\cdots,x_n]/(x_i^2-x_i)_{i=1}^n.$ \\
We get
\[C(x_1,\cdots,x_n)=\sum_{\mathbf{d}\in\mathbb{F}_2^n}A_\mathbf{d} y_\mathbf{d}(x_1,x_2,\cdots,x_n)\]
 where $A_\mathbf{d}=\sum_{\mathbf{c}\in C}(-1)^{\mathbf{c}\cdot \mathbf{d}}$.

We introduce a weighted 2-perfect $\pi$-code with respect to a weighted Hamming metric $d_\pi$.

\begin{definition}\cite{O02}
We say a function $f:\mathbb{F}_2^n\rightarrow \mathbb{R}$ a weighted 2-perfect $\pi$-code if\\
\[\sum_{\mathbf{y}\in S_\pi({\mathbf{x}};2)}f(\mathbf{y})=1 \textrm{ for any } \mathbf{x}\in \mathbb{F}_2^n.\]

\end{definition}

For a given weighted 2-perfect $\pi$-code $f$, a polynomial $f(x_1,\cdots,x_n)$ can be viewed as

\[f(x_1,\cdots,x_n)=\sum_{\mathbf{t}\in \mathbb{F}_2^n}f(\mathbf{t})x_1^{t_1} \cdots x_n^{t_n}. \]

For a 2-perfectness, $|S_\pi(\mathbf{x};2)|=1+m+\binom{m}{2}+(n-m)$ must be power of $2$. Say $2^t$ and let $\mathbf{M}$ be a vector of length $n$ with 
\begin{displaymath}
M_i=\left \{\begin{array}{ll}
1, & \textrm{if 1$\le i\le m$}\\
 0, & \textrm{otherwise}.
\end{array} \right.
\end{displaymath}

\begin{proposition}\label{proposition4.2}

The followings are equivalent:
\begin{enumerate}

\item A function $f$ is a weighted 2-perfect $\pi$-code.
\item $f(x_1,\cdots,x_n)=A_0y_\mathbf{0}(x_1,\cdots,x_n)+\sum_{k=0}^m\sum_{\mathbf{d}\in D_k}{A_\mathbf{d}}(x_1,\cdots,x_n)$ 
\end{enumerate}

where $A_\mathbf{0}=\frac{2^n}{|S_\pi(\mathbf{0};2)|}=2^{n-t}$ and $D_k=\{\mathbf{v}\in \mathbb{F}_2^n$ : $|\mathbf{v}\cap \mathbf{M}|=k, |\mathbf{v}|=2^{t-1}-k(m-k)\}$.

\end{proposition}

\begin{proof}

Let $f$ be a weighted 2-perfect $\pi$-code. \\
Put $f(x_1,\cdots,x_n)=\sum_{\mathbf{d}\in \mathbb{F}_2^n}{A_\mathbf{d}y_\mathbf{d}(x_1,\cdots,x_n)}.$ We show that $A_\mathbf{d}(C)\not=0$ only when $\mathbf{d}=\mathbf{0}$ or $\mathbf{d}\in \cup_{k=0}^{m}D_k$.\\
Note that 
\[\sum_{\mathbf{t}\in \mathbb{F}_2^n}x_1^{t_1}\cdots x_n^{t_n}=(\sum_{\mathbf{d}\in \mathbb{F}_2^n}A_\mathbf{d}y_\mathbf{d}(x_1,\cdots,x_n))\cdot (1+\sum_{i=1}^nx_i+\sum_{1\le i< j\le m}x_ix_j).\]

It follows that for every $\mathbf{d}\not =\mathbf{0}$ , $0=A_\mathbf{d}\{1+n-2|\mathbf{d}|+\binom{m}{2}-2k(m-k)\}$ where $k=|\mathbf{d}\cap \mathbf{M}|$.
Thus, only for $\mathbf{d} =\mathbf{0}$ or $\mathbf{d}\in \cup_{k=0}^mD_k$, $A_\mathbf{d}\not =0$ is possible.

Now, let 

\[f(x_1,\cdots,x_n)=A_\mathbf{0}y_\mathbf{0}(x_1,\cdots,x_n)+\sum_\mathbf{k=0}^m\sum_{\mathbf{d}\in D_k}A_\mathbf{d}y_\mathbf{d}(x_1,\cdots,x_n)\]
 with $A_\mathbf{0}=\frac{2^n}{|S_\pi(\mathbf{0};2)|}$.

 By the comparison of coefficients, we get 
 \[f(\mathbf{v})=\frac{A_\mathbf{0}}{2^n}+\sum_{k=0}^m\sum_{\mathbf{d}\in D_k}\frac{(-1)^\mathbf{d\cdot v}}{2^n}A_\mathbf{d}.\] 
 We need to check if $\sum_{\mathbf{y}\in S_\pi(\mathbf{x};2)}f(\mathbf{y})=1$ for any $\mathbf{x}\in \mathbb{F}_2^n$,

\[\sum_{\mathbf{y}\in S_\pi(\mathbf{x};2)}(\frac{A_\mathbf{0}}{2^n}+\sum_{k=0}^m\sum_{\mathbf{d}\in D_k} \frac{(-1)^\mathbf{d\cdot y}}{2^n}A_\mathbf{d})=1.\]
Since $\sum_{\mathbf{y}\in S_\pi(\mathbf{x};2)}(-1)^\mathbf{d\cdot y}=0$, so $f$ is a weighted 2-perfect $\pi$-code.

\end{proof}

It is easy to find the following corollaries.

\begin{corollary}
Let $C\subset \mathbb{F}_2^n$ be a 2-perfect $\pi$-code. Then,
$\mathbf{u}+\mathbf{c}\in C$ for any $\mathbf{c}\in C$ where $\mathbf{u}=(u_1,\cdots,u_n) \in \mathbb{F}_2^n$ as follows:

\begin{enumerate}
\item If $m$ is odd,
\[u_i=1 \textrm{ for } 1\le i\le n.\]
\end{enumerate}

\begin{enumerate}
\item If $m$ is even,
\begin{displaymath}
u_i=\left \{\begin{array}{ll}
0, & \textrm{if $i\le m$}\\
 1, & \textrm{otherwise}.
\end{array} \right.
\end{displaymath}
\end{enumerate}

\end{corollary}

Let $C\subset \mathbf{F}_2^n$ be a 2-perfect $\pi$-code and $B_\mathbf{u}\cap C:=\{\mathbf{c}\in C:supp(\mathbf{c})\subset supp(\mathbf{u})\}.$
And denote the number of codewords $\mathbf{c}$ of a code $C$ with $|\mathbf{c}\cap \mathbf{M}|=i$ and $|\mathbf{c}\cap \bar{\mathbf{M}}|=j$ by $a_{i,j}$
where $1\le i\le m$ and $1\le j\le n-m.$

\begin{corollary}
Let $C$ be a 2-perfect $\pi$-code. If we have the weight distribution of $B_\mathbf{M}\cap C$, then we have the weight distributions of $B_\mathbf{\bar M}\cap C$ and $C$ with respect to a $\pi$-weight $w_\pi$. $i.e.$ we can get the all values $a_{i,j}$ from $a_{i,0}$ for $i=0,1,\cdots,m.$ 

\end{corollary}

\begin{proof}

We have 
\[
\sum_{i=0}^m a_{i,0}x^i=\sum_{\mathbf{c}\in B_\mathbf{M}\cap C}x^{|c|}=
\frac{1}{2^n} \sum_{\mathbf{v}\in \mathbb{F}_2^n}A_\mathbf{v}(C)(1-x)^{|M\cap \mathbf{v}|}(1+x)^{|M|-|M\cap \mathbf{v}|}\]

 by using the local duality. See \cite{CHK14}.

By Proposition~\ref{proposition4.2}, 
\[\sum_{i=0}^ma_{i,0}x^i=\frac{|C|}{2^n}(1+x)^m+\frac{1}{2^n}\sum_{k=0}^m(\sum_{\mathbf{d}\in D_k}A_\mathbf{d}(C))(1-x)^k(1+x)^{m-k}\cdots (4).\]

\textbf{CLAIM} For $0\le k\le m$, all values $\sum_{\mathbf{d}\in D_k}A_\mathbf{d}(C)$ are determined by $a_{i,0}$.\\
{\it proof of claim)} \\
Substitute $x=1$ into $(4)$ and then we get 

\[\sum_{i=0}^m a_{i,0}=2^{m-t}+\frac{1}{2^{n-m}}\sum_{\mathbf{d}\in D_0}A_d(C).\] 

So, 
\[\sum_{\mathbf{d}\in D_0}A_\mathbf{d}(C)=2^{n-m}\sum_ {i=0}^n a_{i,0}-2^{n-t}.\]

Differentiate $(4)$ $l$ times for $l=1,2,\cdots,m.$ Then, 

\[
\begin{aligned} \sum_{i=l}^m\binom{i}{l}a_{i,0} x^{i-l} &= \frac{1}{2^t} \binom{m}{l} (1+x)^{m-l}\\
& +\frac{1}{2^n}\sum_{k=0}^m (\sum_{\mathbf{d}\in D_k} A_\mathbf{d}(C))\{\sum_{j=0}^l \frac{\partial}{\partial ^jx}(1-x)^k \frac{\partial}{\partial ^{l-j}x}(1+x)^{m-k} \} \cdots (5).
\end{aligned}
\]

Put $x=1$ into $(5)$.
Thus, we get for $0\le l\le m$,

\[ \sum_{k=0}^l  \frac{(-1)^k\binom{m-k}{l-k}}{\binom{l}{k}}(\sum_{\mathbf{d}\in D_k}A_\mathbf{d}(C))= 
2^{n-m+l}~\sum_{i=l}^m\binom{i}{l}a_{i,0} -2^{n-t}\binom{m}{l} .\]

Let $E$ be an $(m+1)\times (m+1)$ matrix with 

\begin{displaymath}
(E)_{i,j}=\left \{\begin{array}{ll}
 \frac{(-1)^j \binom{m-j}{i-j}}{\binom{i}{j}} & \textrm{if $i\ge j$}\\
 0 & \textrm{otherwise}
\end{array} \right.
\end{displaymath}
 where $0\le i,j\le m$.

Let $\mathbf{x}$ and $\mathbf{y}$ be the vectors of length $m+1$ with $x_j=\sum_{\mathbf{d}\in D_j} A_\mathbf{d}(C)$
and $y_j= 2^{n-m+j}\sum_{i=j}^m \binom{i}{j}a_{i,0} 2^{n-t}\binom{m}{j}$, respectively where $x_j$ is the $j$th entry of $\mathbf{x}$ and  $y_j$ is the $j$th entry of $\mathbf{y}$ for $j=0,1,\cdots,m.$

Then, $E\mathbf{x}=\mathbf{y}$ and $E$ is an invertible matrix. Finally, we get the all values $\sum_{\mathbf{d}\in D_k}A_\mathbf{d}(C)$ for $0\le k\le m.$ 

From,

\[
\begin{aligned}
\sum_{\mathbf{c}\in B_{\bar M} \cap C}x^{\mathbf{|c|}}& =\frac{1}{2^n}\sum_{\mathbf{v}\in \mathbb{F}_2^n}A_\mathbf{v}(C)(1-x)^{|\bar M\cap \mathbf{v}|} (1+x)^{n-m-|\bar M\cap \mathbf{v}|}\\
      &=\frac{1}{2^t}(1+x)^{n-m} +\frac{1}{2^n}\sum_{k=0}^m \{ (\sum_{\mathbf{d}\in D_k}A_\mathbf{d}(C))\\
      & \times (1-x)^{2^{t-1}-k(m-k+1)}(1+x)^{n-m-2^{t-1}+k(m-k+1)}\}
\end{aligned}
\]

and 

\[
\begin{aligned} \sum_{\mathbf{c}\in C} x^{w_\pi(\mathbf{c})} & =\frac{1}{2^n}\sum_{\mathbf{v}\in\mathbb{F}_2^n}A_\mathbf{v}(C)(1-x)^{|M\cap \mathbf{v}|} (1+x)^{m-|M\cap \mathbf{v}|} (1-x^2)^{|\bar M\cap \mathbf{v}|} \\
& \times (1+x^2)^{n-m-|\bar M\cap \mathbf{v}|} \\
& =\frac{1}{2^t}(1+x)^m(1+x^2)^{n-m} +\frac{1}{2^n}\sum_{k=0}^m \{(\sum_{\mathbf{d}\in D_k}A_\mathbf{d}(C))(1-x)^k(1+x) ^{m-k} \\
& \times (1-x^2)^{2^{t-1}-k(m-k+1)}(1+x^2)^{n-m-2^{t-1}+k(m-k+1)}\},
\end{aligned}
\]

we get the distributions of $B_{\bar M}\cap C$ and $C$, respectively.

It follows from 
\[\sum_{\mathbf{c}\in C} (\prod_{i=1}^m x_i^{c_i}  \cdot \prod_{i=m+1}^n x_j^{2c_j}) \]
\[=\frac{1}{2^n}\sum_{\mathbf{v} \in \mathbb{F}_2^n} \sum_{\mathbf{c} \in C} (-1)^{\mathbf{v\cdot c}} \prod_{i=1}^m (1-x_i)^{v_i} (1+x_i)^{1-v_j} \prod_{j=m+1}^n (1-x_i^2)^{v_i} (1+x_i^2)^{1-v_j}.\]

\end{proof}

\begin{remark}
In conclusion, we could note that a perfect code in the weighted Hamming metric has similar properties of a perfect code in the Hamming metric. We think there might be much more other similar properties.

\end{remark}


\end{document}